\documentclass[hidelinks,onefignum,onetabnum]{siamonline250211}

\usepackage{amsmath,amssymb}
\usepackage{mathtools}
\usepackage{bm}
\usepackage{booktabs}
\usepackage{colortbl}
\usepackage{graphicx}
\usepackage{array}
\graphicspath{{figures/}}

\newcommand{\R}{\mathbb{R}}
\newcommand{\N}{\mathbb{N}}
\newcommand{\Z}{\mathbb{Z}}

\newcommand{\E}{\mathbb{E}}
\newcommand{\Prb}{\mathbb{P}}

\newcommand{\bx}{\boldsymbol{x}}

\newcommand{\bp}{\boldsymbol{p}}
\newcommand{\bX}{\boldsymbol{X}}
\newcommand{\bpsi}{\boldsymbol{\psi}}
\newcommand{\bPsi}{\boldsymbol{\Psi}}
\newcommand{\bzeta}{\boldsymbol{\zeta}}
\newcommand{\bbeta}{\boldsymbol{\beta}}
\newcommand{\bnu}{\boldsymbol{\nu}}

\newcommand{\bA}{\mathbf{A}}
\newcommand{\bG}{\mathbf{G}}
\newcommand{\bK}{\mathbf{K}}
\newcommand{\bDelta}{\boldsymbol{\Delta}}

\newcommand{\calX}{\mathcal{X}}
\newcommand{\calK}{\mathcal{K}}
\newcommand{\calA}{\mathcal{A}}
\newcommand{\calF}{\mathcal{F}}
\newcommand{\calS}{\mathcal{S}}

\newcommand{\spn}{\operatorname{span}}

\newsiamremark{remark}{Remark}
\newsiamthm{defn}{Definition}
\newsiamthm{claim}{Claim}

\headers{Koopman Operator Inference for the Chemical Master Equation}{A. Dendukuri and L. R. Petzold}
\title{Koopman Operator Inference for the Chemical Master Equation}

\author{Aditya Dendukuri%
\thanks{Department of Computer Science, University of California Santa Barbara
(\email{adendukuri@ucsb.edu}).}
\and Linda R. Petzold%
\thanks{Departments of Mechanical Engineering and Computer Science,
University of California Santa Barbara
(\email{petzold@engineering.ucsb.edu}).}}

\ifpdf
\hypersetup{
hidelinks,
pdftitle={Koopman Operator Inference for the Chemical Master Equation},
pdfauthor={Aditya Dendukuri and Linda R. Petzold}}
\fi

\begin{document}

\maketitle

\begin{abstract}
We infer kinetic rates of stochastic reaction networks from molecular-state
pairs separated by a fixed time lag, without constructing the large state
space of the Chemical Master Equation. Extended dynamic mode decomposition
estimates the finite-time evolution of selected functions of the molecular
state. A matrix logarithm then gives a projected generator whose species rows
determine the rates through the known reaction stoichiometry. A finite set of
functions introduces a systematic bias. We construct additional functions
from the reaction channels and prove that each level of closure moves the
first bias term to higher order. Spectral formulas identify the omitted modes
responsible for this bias and show when the logarithm amplifies sampling
error. Exact calculations and experiments on five reaction networks confirm
the predicted error behavior and show when the added functions improve
finite-sample rate recovery.
\end{abstract}

\begin{keywords}
Koopman operator, extended dynamic mode decomposition, chemical master equation, stochastic reaction network, matrix logarithm, reaction rate inference
\end{keywords}

\begin{MSCcodes}
65F60, 37M25, 92C42, 60J27
\end{MSCcodes}

% =============================================================================
\section{Introduction}
\label{sec:intro}
% =============================================================================

The Chemical Master Equation (CME) governs the probability distribution of
molecular copy numbers in a well-mixed stochastic reaction
network~\cite{gillespie1977exact,mcquarrie1967stochastic}. On a finite state
space, the probability vector $p(t)$ satisfies
\[
 \frac{d}{dt}p(t)=\bA(c)p(t),
\]
where the reaction stoichiometry, propensities, and rate vector $c$ determine
the generator $\bA(c)$~\cite{gillespie1992rigorous}. Our inverse problem is to
recover $c$ from pairs of molecular states observed at a fixed time
separation.

The molecular state space grows rapidly with the number of species and their
copy numbers. A likelihood based on the full CME requires a transition model
whose dimension is the number of retained states. Finite-state projection
reduces this cost by truncation~\cite{munsky2006finite}. Moment methods replace
the distribution by a finite set of moments, but generally require a closure
approximation~\cite{engblom2006computing,schnoerr2017approximation}. We instead
pose the inverse problem on selected functions of the molecular state. The
number of these observables, which we also call the dictionary, need not grow
with the number of molecular states.

Let $\bX(t)$ denote the vector of molecular copy numbers at time $t$, and let $\bx$ denote a possible initial state. An \emph{observable} is a function $f$ of these copy numbers. Its conditional expected value after time $t$ is
\[
 (\calK_t f)(\bx)=\E[f(\bX(t))\mid\bX(0)=\bx].
\]
The operators $\{\calK_t\}_{t\geq0}$ form the Koopman semigroup. The instantaneous rate of change of an observable is given by
\[
 (\calA f)(\bx)
 =\lim_{t\downarrow0}\frac{(\calK_t f)(\bx)-f(\bx)}{t}.
\]
This operator $\calA$ is the Koopman generator. For the CME,
$\calA=\bA(c)^\top$. From observations separated by a time $h$, extended
dynamic mode decomposition (EDMD) approximates $\calK_h$ on $d$
observables~\cite{williams2015data}. Its principal matrix logarithm
approximates $\calA$, and the rows associated with the species determine the
reaction rates. The fitted matrix has dimension $d$, independent of the
number of molecular states.

The estimator has three stages:
\begin{equation}\label{eq:intro_estimator}
\begin{aligned}
 \widehat{\bK}_h&=\bPsi_Y\bPsi_X^\dagger
 &&\text{(finite-time Koopman matrix)},\\
 \widehat{\bA}_h&=h^{-1}\log\widehat{\bK}_h
 &&\text{(Koopman generator matrix)},\\
 \widehat c&=\mathcal R(\widehat{\bA}_h)
 &&\text{(reaction rates)}.
\end{aligned}
\end{equation}
Here $\bPsi_X$ and $\bPsi_Y$ contain the selected observables evaluated at the two states in each pair, and $\mathcal R$ extracts rates from the species rows. The fitted square matrices have dimension $d$, the number of selected observables.

Prior work uses Koopman lifting and the logarithm of an EDMD matrix to recover governing equations from snapshot pairs ~\cite{mauroy2020koopman,drmac2021koopmanlog}. Other work uses generator EDMD to approximate the infinitesimal operator directly and identify drift and diffusion coefficients~\cite{klus2020data}. Stochastic-model parameters can also be fitted to an estimated operator~\cite{riseth2018operator}, while estimates at several observation intervals can be combined to recover a Markov-process generator ~\cite{sechi2021newton}.

For Markov jump processes, Niemann, Klus, and Sch\"utte recover transition-rate combinations from a Koopman-generator model of the large-population diffusion limit~\cite{niemann2021agent}. Their construction uses short-time Kramers--Moyal estimates of drift and diffusion. Classical generator-inference methods fit a state-to-state rate matrix from a discretely observed transition matrix~\cite{crommelin2006fitting}. We restrict the exact jump-process generator to selected mass-action observables, so the matrix dimension equals the dictionary dimension rather than the number of molecular states.

Gupta and Khammash approximate the continuous-time Koopman operator of a specified reaction network through its dominant decay modes ~\cite{gupta2025srnkoopman}. They use this representation to predict moments across initial states and compute parameter sensitivities and spectral densities. In our inverse problem, we determine unknown reaction rates from state pairs at one fixed separation through \eqref{eq:intro_estimator}. The finite-dimensional projection therefore occurs before logarithmic inversion and can bias the rates even when the projected finite-time operator is known exactly.

Mass-action kinetics provides a natural observable set. If
$\phi_{\bnu_k}$ is the propensity monomial for reaction $k$, then
\begin{equation}\label{eq:intro_species_action}
 \calA x_s=\sum_{k=1}^M \zeta_{s,k}c_k\phi_{\bnu_k}.
\end{equation}
The coefficient of each propensity is therefore a stoichiometric change multiplied by a reaction rate. Once the observable set contains these monomials, identifiable rates can be read from the species rows of the exact Koopman generator.

The same structure gives an observable hierarchy. The generator is affine in
the rates,
\begin{equation}\label{eq:intro_channel_expansion}
 \calA(c)=\sum_{k=1}^M c_k\calA_k,
 \qquad
 (\calA_k f)(\bx)
 =\phi_{\bnu_k}(\bx)
   \bigl[f(\bx+\bzeta_k)-f(\bx)\bigr],
\end{equation}
and every unit-rate channel operator $\calA_k$ is known once the reaction set is specified. Repeated generator images therefore have the expansion
\[
 \calA(c)^j x_s
 =\sum_{k_1,\ldots,k_j=1}^M
 c_{k_j}\cdots c_{k_1}
 \calA_{k_j}\cdots\calA_{k_1}x_s.
\]
At any fixed depth, all functions on the right can be generated without the
unknown rates. They form a reaction-defined closure hierarchy for the species
rows that contain the physical parameters. The falling-factorial basis makes
the growth of this hierarchy explicit. If $\mathcal P_q$ denotes polynomials
of degree at most $q$ and reaction $k$ has order $o_k=|\bnu_k|$, then
\begin{equation}\label{eq:intro_degree_growth}
 \calA_k\mathcal P_q\subseteq\mathcal P_{q+o_k-1}.
\end{equation}
The jump difference lowers degree before multiplication by the propensity.
Polynomial spaces are therefore invariant for networks with only zero- and
first-order reactions, and the projected logarithm is exact for sufficiently
small admissible $h$. For nonlinear reactions,
\eqref{eq:intro_degree_growth} bounds the new degrees introduced at each
closure depth. This gives exact closure for linear reaction networks and a
controlled, rate-independent enlargement for nonlinear networks.

Finite-time data determine a projected operator. If $P$ projects onto the
selected observable space $\calF$, then exact finite-time expectations give
the generator
\begin{equation}\label{eq:intro_projection_defect}
 h^{-1}\log\!\left(Pe^{h\calA}|_{\calF}\right),
 \qquad
 h^{-1}\log\!\left(Pe^{h\calA}|_{\calF}\right)
 \ne P\calA|_{\calF}
\end{equation}
in general. We therefore obtain biased species rows even before sampling error enters.

We reduce this error by adding successive generator images of the species
counts used for rate extraction. An observable space has \emph{closure depth
$r$} when it contains $\calA x_s,\ldots,\calA^r x_s$ for every selected
species $x_s$. This definition links the reaction-network identity directly
to the finite-time error in the physical rates.

The analysis gives a species-row identifiability condition, a modal-residue
description of the omitted Koopman modes, and a rate-bias theorem at every
closure depth. We also propagate sampling error through the logarithm. Exact
finite-state matrices recover the predicted bias orders, and independent
trajectory data measure the effect of the added observables on four further
networks. Throughout, the reaction set is known, observations occur in
pairs, and the selected species coefficients identify the rates. The unknowns
are the rates within this fixed reaction set.

\section{Reaction Networks and the Chemical Master Equation}
\label{sec:networks}
% =============================================================================

We first describe the reaction network and its two generators. The forward
generator evolves probabilities, while the backward generator evolves the
observables used to recover the rates.

\subsection{Reaction network model}

Consider a network of $N$ molecular species and $M$ reaction channels.
\begin{defn}[Stochastic reaction network]
\label{def:srn}
A stochastic reaction network on $N$ species $\calS_1,\dots,\calS_N$ consists of $M$ reaction channels indexed by $k=1,\dots,M$, each specified by:
\begin{itemize}
\item a stoichiometric vector $\bzeta_k=\bbeta_k-\bnu_k\in\Z^N$, where $\nu_{ik},\beta_{ik}\in\N$ are the reactant and product stoichiometric coefficients for species $\calS_i$; and
\item a propensity function $\alpha_k:\mathbb N_0^N\to[0,\infty)$, where, conditional on state $\bx$ at time $t$, the probability that reaction $k$ fires in $[t,t+dt)$ is $\alpha_k(\bx)\,dt+o(dt)$.
\end{itemize}
The state $\bx=(x_1,\dots,x_N)^\top\in\mathbb N_0^N$ records the molecular copy numbers, and $c=(c_1,\dots,c_M)^\top$ denotes the vector of reaction rates.
\end{defn}

Under mass-action kinetics, reaction $k$ has propensity
\begin{equation}\label{eq:mass_action}
\alpha_k(\bx)=c_k\phi_{\bnu_k}(\bx),\qquad
\phi_{\bnu_k}(\bx)\coloneqq\prod_{i=1}^N\binom{x_i}{\nu_{ik}},
\end{equation}
where $c_k>0$ is the rate constant and the binomial factor counts the possible
reactant groups. We call $\phi_{\bnu_k}$ the \emph{propensity monomial} of
reaction $k$. The reactant stoichiometry fixes this monomial, while $c_k$ is
unknown. A source reaction has monomial $1$, and a unimolecular reaction in
species $i$ has monomial $x_i$. A bimolecular reaction has either $x_ix_j$ or
$x_i(x_i-1)/2$; trimolecular reactions give the analogous degree-three
factors.

\subsection{Chemical master equation}
\label{sec:cme}

The reaction channels define a continuous-time Markov chain and its CME
generator.
\begin{defn}[CME generator]
\label{def:cme}
The probability mass function $p(\bx,t)=\Prb[\bX(t)=\bx]$ satisfies the \emph{Chemical Master Equation}~\cite{ethier1986markov,gillespie1992rigorous}
\begin{equation}\label{eq:cme}
\frac{d}{dt}p(\bx,t)
=\sum_{k=1}^M\bigl[\alpha_k(\bx-\bzeta_k)p(\bx-\bzeta_k,t)
-\alpha_k(\bx)p(\bx,t)\bigr].
\end{equation}
On a finite closed state set $\calX=\{\bx_1,\dots,\bx_n\}$, let $\bp(t)=(p(\bx_1,t),\ldots,p(\bx_n,t))^\top$. Then $\dot{\bp}(t)=\bA(c)\bp(t)$, where $\bA(c)\in\R^{n\times n}$ has entries
\[
\bA(c)_{\bx_i,\bx_j}
 =\sum_{k:\,\bx_i=\bx_j+\bzeta_k}\alpha_k(\bx_j)\quad(i\ne j),
 \qquad
\bA(c)_{\bx_j,\bx_j}=-\sum_{k:\,\bzeta_k\ne0}\alpha_k(\bx_j),
\]
and all columns sum to zero. The state set is closed if every enabled reaction remains in $\calX$. An FSP truncation adds a sink state for transitions that leave its retained set.
\end{defn}

An inverse problem posed directly in terms of $\bA(c)$ uses probability vectors of dimension $n=|\calX|$, which can grow combinatorially with the copy numbers. The Koopman formulation follows selected conditional expectations.

\section{From Koopman Observables to Reaction Rates}
\label{sec:koopman}

The Koopman generator acts on functions of the molecular state rather than on
state probabilities. For mass-action networks, its action on a species count
is a known linear combination of the reaction rates. This identity places the
unknown rates in selected rows of a small generator matrix.

\subsection{Koopman semigroup and generator}

\begin{defn}[Koopman semigroup]
\label{def:koopman}
An \emph{observable} is a measurable function $f:\calX\to\R$. For bounded observables $f$, let the \emph{Koopman semigroup} $\{\calK_t\}_{t\geq 0}$ act by
\begin{equation}\label{eq:koopman_semigroup}
(\calK_t f)(\bx)=\E\bigl[f(\bX(t))\mid\bX(0)=\bx\bigr].
\end{equation}
\end{defn}

The semigroup satisfies $\calK_0=I$ and $\calK_{t+s}=\calK_t\calK_s$ for all $t,s\geq 0$. It is the evolution of observable functions dual to the probability evolution generated by $\bA(c)$.

\begin{defn}[Koopman generator]
\label{def:koopman_gen}
Let the \emph{Koopman generator} denote the instantaneous rate of change of an observable under the Koopman semigroup:
\begin{equation}\label{eq:koopman_generator_limit}
(\calA f)(\bx)
\coloneqq\lim_{t\downarrow0}
\frac{(\calK_t f)(\bx)-f(\bx)}{t},
\end{equation}
whenever this limit exists. For a stochastic reaction network, $\calA=\bA(c)^\top$, and we evaluate the limit to obtain
\begin{equation}\label{eq:koopman_generator}
(\calA f)(\bx)=\sum_{k=1}^M\alpha_k(\bx)\bigl[f(\bx+\bzeta_k)-f(\bx)\bigr]
\end{equation}
for any bounded observable $f$. It follows that
$\frac{d}{dt}\E[f(\bX(t))]=\E[(\calA f)(\bX(t))]$, which describes the
evolution of expected observables. For unbounded polynomial observables,
these identities require the corresponding moment and generator-domain
conditions.
\end{defn}

\begin{remark}[Observable-space restriction]\label{rem:observable_space}
On a finite state space, the forward CME generator and backward Koopman generator are dual:
\begin{equation}\label{eq:forward_backward_duality}
\begin{aligned}
 \langle f,\bA(c)p\rangle
   &=\langle \bA(c)^\top f,p\rangle,\\
 \calA&=\bA(c)^\top,\\
 \E_{p_0}[f(\bX(t))]
   &=f^\top e^{t\bA(c)}p_0
    =p_0^\top e^{t\calA}f.
\end{aligned}
\end{equation}
On a countable state space, the same dual relation holds when the generators and pairing have compatible domains, and both spaces may be infinite dimensional. The finite model used here is the restriction of $\calK_t$ to the observables needed for rate recovery. Its dimension is the number of retained functions; its approximation error is analyzed in Sections~\ref{sec:method}--\ref{sec:block_analysis}.
\end{remark}

\subsection{Mass-action basis}
\label{sec:mab}

\begin{defn}[Mass-action basis]\label{def:mab}
For a multi-index $\boldsymbol\gamma=(\gamma_1,\dots,\gamma_N)\in\mathbb{N}_0^N$, the corresponding \emph{mass-action monomial} is
\begin{equation}\label{eq:mab_monomial}
\phi_{\boldsymbol\gamma}(\bx)
=\prod_{i=1}^N\binom{x_i}{\gamma_i}.
\end{equation}
For an order $q$ at least as large as the largest reaction order, the \emph{mass-action dictionary}
\begin{equation}\label{eq:mab_propensity}
 \bpsi(\bx)
 =\bigl(\psi_1(\bx),\dots,\psi_d(\bx)\bigr)^\top
\end{equation}
contains the monomials $\phi_{\boldsymbol\gamma}$ with $|\boldsymbol\gamma|\leq q$, ordered so that the constant and species coordinates appear first. Before dependencies are removed, it contains $d=\binom{N+q}{q}$ functions. We retain an independent subset on the state space under study while preserving the constant, the species, and a representation of every propensity monomial.
\end{defn}

Before conservation reduction, each reaction propensity has the coordinate representation
\begin{equation}\label{eq:propensity_coordinate}
 \alpha_k(\bx)=c_k\psi_{\ell(k)}(\bx),\qquad k=1,\dots,M,
\end{equation}
where $\ell(k)$ is the dictionary index for $\phi_{\bnu_k}=\psi_{\ell(k)}$. After conservation reduction, the same monomial may instead be represented by a linear combination of dictionary functions. The corresponding extraction formula follows the basic single-coordinate case.

\subsection{Generator applied to species counts}
\label{sec:closure}

Let $\calF=\spn\{\psi_1,\dots,\psi_d\}$ be the subspace spanned by the mass-action basis in Definition~\ref{def:mab}. For each reaction channel, let the unit-rate reaction operator be
\begin{equation}\label{eq:channel_generator}
 (\calA_k f)(\bx)
 =\phi_{\bnu_k}(\bx)\bigl[f(\bx+\bzeta_k)-f(\bx)\bigr].
\end{equation}
The full Koopman generator is therefore
$\calA=\sum_{k=1}^M c_k\calA_k$. The specified reactions determine every
$\calA_k$ without knowledge of the rates.

The classical mass-action first-moment identity
\cite{gillespie1992rigorous} now gives the coefficient rows used for rate
identification and for the analysis of the projected matrix logarithm.
\begin{lemma}[Generator applied to a species count]\label{clm:species_action}
Let $\bpsi$ be the mass-action basis of Definition~\ref{def:mab}. Then, for each species $s$,
\begin{equation}\label{eq:species_closure}
(\calA x_s)(\bx)=\sum_{k=1}^M
 \zeta_{s,k}c_k\phi_{\bnu_k}(\bx).
\end{equation}
Here $\zeta_{s,k}$ is the $s$th component of $\bzeta_k$. In particular, $\calA x_s\in\calF$, and the expected species evolution satisfies
\begin{equation}\label{eq:mean_evolution}
\frac{d}{dt}\E[x_s(\bX(t))]
=\sum_{k=1}^M\zeta_{s,k}c_k
 \E[\phi_{\bnu_k}(\bX(t))]
\end{equation}
exactly at the species level.
\end{lemma}

\begin{proof}
A firing of channel $k$ changes $x_s$ by $\zeta_{s,k}$. We multiply this increment by its propensity and sum over channels to obtain both identities.
\end{proof}

\paragraph{Matrix convention}
Let $P$ denote projection onto $\calF$, and let $\bA$ be the $d\times d$ coordinate matrix of $P\calA|_{\calF}$ in the ordered physical dictionary $\bpsi$, with coefficients stored by row. Thus $(\bA)_{i,j}$ is the coefficient of $\psi_j$ in $P\calA\psi_i$. The species coordinates are retained explicitly and, by Lemma~\ref{clm:species_action}, $P\calA x_s=\calA x_s$. Hence, the corresponding rows of $\bA$ contain exact linear combinations of the rates. Later applications of the generator can leave $\calF$, and the resulting projection error is the subject of Section~\ref{sec:block_analysis}.

\paragraph{Species-row structure}
With rows labeled by the observable acted upon and columns by the resulting dictionary coordinate, the matrix has the form
\begin{equation}\label{eq:koopman_matrix_anatomy}
\bA=
\bordermatrix{
 & \psi_1 & \cdots & \psi_j & \cdots & \psi_d \cr
 1   & 0 & \cdots & 0 & \cdots & 0 \cr
 x_1 & (\bA)_{1,1} & \cdots & (\bA)_{1,j} & \cdots & (\bA)_{1,d} \cr
 \vdots & \vdots & \ddots & \vdots & & \vdots \cr
 x_s & (\bA)_{s,1} & \cdots & (\bA)_{s,j} & \cdots & (\bA)_{s,d} \cr
 \vdots & \vdots & & \vdots & \ddots & \vdots \cr
 x_N & (\bA)_{N,1} & \cdots & (\bA)_{N,j} & \cdots & (\bA)_{N,d} \cr
 \text{higher observables}
     & \ast & \cdots & \ast & \cdots & \ast
 }.
\end{equation}
The constant row is zero because $\calA1=0$. For a species row, the generator action can be collected by dictionary column:
\begin{align}
 P\calA x_s=\calA x_s
 &=\sum_{k=1}^M \zeta_{s,k}c_k\psi_{\ell(k)} \notag\\
 &=\sum_{j=1}^d
 \left(\sum_{k:\,\ell(k)=j}\zeta_{s,k}c_k\right)\psi_j.
 \label{eq:species_row_columns}
\end{align}
Consequently,
\begin{equation}\label{eq:koopman_species_entries}
 (\bA)_{s,j}
 =\sum_{k:\,\ell(k)=j}\zeta_{s,k}c_k,
 \qquad s=1,\ldots,N,\quad j=1,\ldots,d.
\end{equation}
Equivalently, each propensity column has the species block
\begin{equation}\label{eq:koopman_species_column}
 \begin{bmatrix}
  (\bA)_{1,j}\\
  \vdots\\
  (\bA)_{N,j}
 \end{bmatrix}
 =\sum_{k:\,\ell(k)=j}c_k\bzeta_k.
\end{equation}
Reaction $k$ therefore contributes the scaled stoichiometric vector
$c_k\bzeta_k$ to its propensity column $\ell(k)$. The starred rows describe
the action on higher observables and need not close in the retained
dictionary.

When several reactions share a propensity function, one coefficient no
longer determines one rate. The species rows still separate the rates when
their stoichiometric changes are linearly independent. The following
corollary states this condition and defines the extraction map.
\begin{corollary}[Rate extraction]\label{cor:rates}
Suppose reactions $k_1,\ldots,k_g$ share the propensity observable $\psi_j$. Their rates and the corresponding species-column coefficients satisfy
\begin{equation}\label{eq:generator_structure}
 \begin{bmatrix}
  (\bA)_{1,j}\\ \vdots\\ (\bA)_{N,j}
 \end{bmatrix}
 =
 \begin{bmatrix}
  \zeta_{1,k_1}&\cdots&\zeta_{1,k_g}\\
  \vdots&&\vdots\\
  \zeta_{N,k_1}&\cdots&\zeta_{N,k_g}
 \end{bmatrix}
 \begin{bmatrix}c_{k_1}\\ \vdots\\ c_{k_g}\end{bmatrix}.
\end{equation}
The rates are identifiable from this column if and only if the displayed stoichiometric matrix has full column rank. In that case, we recover the rates by multiplying the species-column vector by the pseudoinverse.
\end{corollary}
\begin{proof}
Collect the coefficients of the shared propensity in \eqref{eq:species_closure}. The resulting linear system has a unique rate vector exactly when its stoichiometric matrix has full column rank.
\end{proof}

Let $\mathcal R$ denote the fixed linear map obtained by collecting all full-rank instances of~\eqref{eq:generator_structure} from the species rows to the rate vector: $\mathcal R(\bA)=c$. This notation also covers conservation-reduced dictionaries, where a propensity is represented by a known linear combination of columns. The rank condition is verified for every numerical experiment; we may use additional observable rows to identify rates that the species rows alone cannot separate.

Species-row closure is weaker than invariance of $\calF$ under $\calA$. If channel $k$ has reaction order $o_k$, the jump difference lowers the observable degree by at least one, so
\[
 \calA_k\mathcal P_q\subseteq\mathcal P_{q+o_k-1},\qquad \calA_k1=0,
\]
where $\mathcal P_q$ contains polynomials of degree at most $q$. Thus, repeated application of the reaction operators for bimolecular and trimolecular channels can produce functions outside a fixed dictionary. The local expansions are therefore proved on a finite closed state space. Their extension to an unbounded network would require additional domain and moment bounds; the unbounded examples in Section~\ref{sec:results} lie outside the stated scope of the theorem.

\section{Estimating the Finite-Time Koopman Operator}
\label{sec:method}
\label{sec:error}
% =============================================================================

Each observation is a pair $(\bX(t_i),\bX(t_i+h))$ separated by the same
interval $h$. EDMD fits the finite-time Koopman operator from these pairs. We
take its principal matrix logarithm, extract the rates, and then propagate the
sampling error through these calculations.

For $m$ paired observations, arrange the evaluated observables as columns:
\begin{equation}\label{eq:feature_matrices}
\begin{aligned}
 \bPsi_X
 &=\begin{bmatrix}
 \bpsi(\bX(t_1))&\cdots&\bpsi(\bX(t_m))
 \end{bmatrix},\\
 \bPsi_Y
 &=\begin{bmatrix}
 \bpsi(\bX(t_1+h))&\cdots&\bpsi(\bX(t_m+h))
 \end{bmatrix}.
\end{aligned}
\end{equation}
When $\bPsi_X$ has full row rank, the empirical moment matrices and their induced operator are
\begin{equation}\label{eq:compact_estimator}
\begin{aligned}
 \widehat{\bG}_0&=\frac1m\bPsi_X\bPsi_X^\top,\\
 \widehat{\bG}_1&=\frac1m\bPsi_Y\bPsi_X^\top,\\
 \widehat{\bK}_h
   &=\widehat{\bG}_1\widehat{\bG}_0^{-1}
     =\bPsi_Y\bPsi_X^\dagger,\\
 \widehat{\bA}_h&=h^{-1}\log\widehat{\bK}_h,\\
 \widehat c&=\mathcal R(\widehat{\bA}_h).
\end{aligned}
\end{equation}
The dictionary contains the constant, selected species, and every mass-action propensity required for extraction. Fixed feature scaling is undone before the linear map $\mathcal R$ is applied. Rate estimates require full dictionary rank and a real principal logarithm.

\subsection{Infinite-sample limit}

Let $\mu$ be the law of the current state, possibly a mixture of transient observation laws. The population moment matrices have entries
\begin{equation}\label{eq:gram_pop_compact}
\begin{aligned}
 (\bG_0)_{ij}
 &=\E_\mu[\psi_i(\bX(t))\psi_j(\bX(t))],\\
 (\bG_1)_{ij}
 &=\E_\mu[\psi_i(\bX(t+h))\psi_j(\bX(t))].
\end{aligned}
\end{equation}
The infinite-sample, or population, EDMD matrix is $\bK_h=\bG_1\bG_0^{-1}$. It represents the $L^2(\mu)$ projection of $\calK_h$ onto the dictionary span. For independent trajectories, or for one stationary ergodic trajectory, we apply entrywise laws of large numbers to obtain $\widehat{\bG}_i\to\bG_i$. If $\bG_0$ is invertible and $\bK_h$ is admissible, we use continuity of inversion and of the principal logarithm to obtain
\begin{equation}\label{eq:Ahat_conv_compact}
 \widehat c=\mathcal R(h^{-1}\log\widehat{\bK}_h)
 \longrightarrow \mathcal R(h^{-1}\log\bK_h)=c_h.
\end{equation}
Paired data therefore recover $\bK_h$ and $c_h$ in the infinite-sample limit.
Increasing the sample size at fixed $h$ removes sampling error but leaves the
projection bias $c_h-c$. This convergence requires finite expectations for
the Gram-matrix entries and applies to independent trajectories or to a
stationary ergodic trajectory.

\paragraph{Sampling perturbations.}
Let $\bDelta_i=\widehat{\bG}_i-\bG_i$. At fixed invertible $\bG_0$, we use first-order matrix perturbation to obtain
\begin{align}
 \widehat{\bG}_0^{-1}
 &=\bG_0^{-1}-\bG_0^{-1}\bDelta_0\bG_0^{-1}
   +o_p(m^{-1/2}),\label{eq:gram_inverse_perturbation}\\
 \widehat{\bK}_h-\bK_h
 &=\bDelta_1\bG_0^{-1}
   -\bG_1\bG_0^{-1}\bDelta_0\bG_0^{-1}
   +o_p(m^{-1/2})\notag\\
 &=(\bDelta_1-\bK_h\bDelta_0)\bG_0^{-1}
   +o_p(m^{-1/2}).\label{eq:koopman_perturbation}
\end{align}
Applying the Fr\'echet derivative of the matrix logarithm and then the rate-extraction map, we obtain
\begin{equation}\label{eq:compact_sampling_chain}
 \widehat c-c_h
 =h^{-1}\mathcal R\!\left(D\log(\bK_h)\left[
 (\bDelta_1-\bK_h\bDelta_0)\bG_0^{-1}\right]\right)
 +o_p(m^{-1/2}).
\end{equation}
The covariance of the two Gram matrices is therefore propagated jointly through inversion, the matrix logarithm, and stoichiometric extraction. For $m$ independent trajectory contributions with finite second moments, the multivariate central limit theorem and delta method make the left side asymptotically normal, with covariance obtained from the displayed linear map. For independently sampled pairs, its short-interval rate variance is $\nu/h+O(1)$; hence the standard deviation is of order $(mh)^{-1/2}$. Whole-trajectory resampling is used below when pairs overlap.

\section{Reducing Finite-Time Bias with Reaction-Generated Observables}
\label{sec:block_analysis}

Projection introduces bias because an observable can leave the retained space
and return before the next observation. We describe these excursions in two
ways. A block expansion gives their order in the observation interval, while
a resolvent formula identifies the omitted Koopman modes that carry them.
Reaction-generated observables remove the shortest excursions from the
species rows used for rate extraction.

\subsection{Projection bias in the species rows}

Consider a full-support measure on a finite closed state space and let $P$ be the orthogonal projection onto the independent dictionary span $\calF$. Let
\begin{equation}\label{eq:matrix_chain}
\begin{aligned}
 \bK_h&\coloneqq P e^{h\calA}|_{\calF},\\
 \bA_h&\coloneqq h^{-1}\log \bK_h,\\
 c_h&\coloneqq\mathcal R(\bA_h).
\end{aligned}
\end{equation}
Here $\bK_h$ is the projected finite-time action, $\bA_h$ is its logarithmic
generator, and $c_h$ is the corresponding infinite-sample rate estimate. A
hat marks an empirical quantity. We call a real projected matrix
\emph{admissible} if it has no eigenvalues on $(-\infty,0]$, so that its real
principal logarithm exists. An observation interval is admissible when its
projected matrix is admissible.

Choose bases for $\calF$ and $\calF^\perp$ that are orthonormal in this measure. Together, these bases are adapted to $\calF\oplus\calF^\perp$. In the row-coefficient convention of EDMD, the full generator has the block form
\begin{equation}\label{eq:A_blocks}
 \bA_{\rm full}=
 \begin{bmatrix}
 \bA_{JJ}&\bA_{JJ'}\\
 \bA_{J'J}&\bA_{J'J'}
 \end{bmatrix},\qquad \bK_h=(e^{h\bA_{\rm full}})_{JJ}.
\end{equation}
Let $J$ index the retained observables and $J'$ their complement. The block
$\bA_{JJ'}$ contains the coefficients outside $\calF$ that appear when the
generator acts on a retained observable. The matrix $\bK_h$ is the
row-convention representation of $P e^{h\calA}|_{\calF}$, and the subscript
``sp'' selects the species rows used by $\mathcal R$. A change of retained
basis is absorbed into $\mathcal R$, so it does not change the extracted
rates.

\begin{lemma}[Projected Koopman operator]\label{lem:projected_operator}
For fixed projection and $h\to0$,
\begin{align}\label{eq:projected_operator_expansion}
 \bK_h={}&I+h\bA_{JJ}
 +\frac{h^2}{2}(\bA_{JJ}^2+\bA_{JJ'}\bA_{J'J})\notag\\
 &+\frac{h^3}{6}\bigl(\bA_{JJ}^3
 +\bA_{JJ}\bA_{JJ'}\bA_{J'J}
 +\bA_{JJ'}\bA_{J'J}\bA_{JJ}
 +\bA_{JJ'}\bA_{J'J'}\bA_{J'J}\bigr)+O(h^4).
\end{align}
\end{lemma}
\begin{proof}
We use block multiplication to obtain the required blocks of the second power:
\begin{align*}
 (\bA_{\rm full}^2)_{JJ}&=\bA_{JJ}^2+\bA_{JJ'}\bA_{J'J},\\
 (\bA_{\rm full}^2)_{J'J}&=\bA_{J'J}\bA_{JJ}+\bA_{J'J'}\bA_{J'J}.
\end{align*}
We multiply once more to obtain
\begin{align*}
 (\bA_{\rm full}^3)_{JJ}
 &=\bA_{JJ}(\bA_{\rm full}^2)_{JJ}
   +\bA_{JJ'}(\bA_{\rm full}^2)_{J'J}\\
 &=\bA_{JJ}^3
   +\bA_{JJ}\bA_{JJ'}\bA_{J'J}
   +\bA_{JJ'}\bA_{J'J}\bA_{JJ}
   +\bA_{JJ'}\bA_{J'J'}\bA_{J'J}.
\end{align*}
Using $e^{h\bA_{\rm full}}=I+h\bA_{\rm full}+h^2\bA_{\rm full}^2/2 +h^3\bA_{\rm full}^3/6+O(h^4)$, its $JJ$ block is
\[
 (e^{h\bA_{\rm full}})_{JJ}
 =I+h\bA_{JJ}
 +\frac{h^2}{2}(\bA_{JJ}^2+\bA_{JJ'}\bA_{J'J})
 +\frac{h^3}{6}(\bA_{\rm full}^3)_{JJ}+O(h^4).
\]
We replace $(\bA_{\rm full}^3)_{JJ}$ by the expression above to obtain \eqref{eq:projected_operator_expansion}.
\end{proof}

\subsection{How projection changes the Koopman spectrum}

We interpret the projection in \eqref{eq:A_blocks} spectrally. If $\calF$ is invariant under $\calA$, then $\bA_{JJ'}=0$ and
\begin{equation}\label{eq:spectral_mapping}
 \bK_h=e^{h\bA_{JJ}},\qquad
 \sigma(\bK_h)=e^{h\sigma(\bA_{JJ})}.
\end{equation}
Thus the fitted Koopman eigenvalues are exponentials of generator eigenvalues. If $\sigma(h\bA_{JJ})$ lies in the open strip $\{z:-\pi<\operatorname{Im}z<\pi\}$, we use the principal logarithm to reverse this mapping and obtain $h^{-1}\log\bK_h=\bA_{JJ}$. Outside that strip, generator frequencies separated by $2\pi/h$ cannot be distinguished. These are the matrix-function form of the sampling and branch conditions for Koopman spectral inference~\cite{higham2008functions,chandrasekaran2016core}.

When the retained observable space is not invariant, paths through $J'$ alter the spectrum. Let $\Gamma$ be a contour in the common resolvent set that encloses the spectra of $\bA_{\rm full}$ and $\bA_{J'J'}$. We use block inversion and the Cauchy functional calculus to obtain
\begin{align}
 \bigl[(zI-\bA_{\rm full})^{-1}\bigr]_{JJ}
 &=\bigl[zI-\bA_{JJ}
 -\bA_{JJ'}(zI-\bA_{J'J'})^{-1}\bA_{J'J}\bigr]^{-1},
 \label{eq:block_resolvent}\\
 \bK_h
 &=\frac{1}{2\pi\mathrm i}\int_\Gamma e^{hz}
 \bigl[zI-\bA_{JJ}
 -\bA_{JJ'}(zI-\bA_{J'J'})^{-1}\bA_{J'J}\bigr]^{-1}\,dz.
 \label{eq:koopman_resolvent}
\end{align}

\begin{proposition}[Modal residues of the omitted space]
\label{prop:modal_residues}
Suppose $\bA_{J'J'}$ is diagonalizable, with right and left eigenvectors $\bm v_\ell$ and $\bm w_\ell$ normalized by $\bm w_\ell^*\bm v_k=\delta_{\ell k}$, and eigenvalues $\mu_\ell$, where $^*$ denotes conjugate transpose. Then
\begin{equation}\label{eq:modal_self_energy}
 \bA_{JJ'}(zI-\bA_{J'J'})^{-1}\bA_{J'J}
 =\sum_\ell\frac{\bm C_\ell}{z-\mu_\ell},
 \qquad
 \bm C_\ell=\bA_{JJ'}\bm v_\ell\bm w_\ell^*\bA_{J'J}.
\end{equation}
Consequently, an omitted mode enters the projected Koopman operator through its decay rate $\mu_\ell$ and coupling residue $\bm C_\ell$. If $\bm C_\ell=0$, that mode is invisible to the retained finite-time operator.
\end{proposition}
\begin{proof}
We substitute the spectral resolution $(zI-\bA_{J'J'})^{-1}=\sum_\ell \bm v_\ell\bm w_\ell^*/(z-\mu_\ell)$ into \eqref{eq:block_resolvent} to obtain~\eqref{eq:modal_self_energy}. We obtain the final statement from~\eqref{eq:koopman_resolvent}.
\end{proof}

The middle term in \eqref{eq:block_resolvent} represents a
frequency-dependent return through the omitted observables. It shows how
unresolved modes perturb the projected Koopman spectrum. When either coupling
block vanishes, \eqref{eq:spectral_mapping} is recovered. This
Schur-complement resolvent is standard in invariant-subspace and
matrix-function analysis~\cite{horn2013matrix,higham2008functions}.

Equation~\eqref{eq:modal_self_energy} also connects the analysis with the
stochastic Koopman approximation of Gupta and
Khammash~\cite{gupta2025srnkoopman}. They estimate iterated resolvents to
recover dominant decay modes of a specified reaction network. Here each pole
$\mu_\ell$ is paired with a residue $\bm C_\ell$ that measures how the mode
leaves and re-enters the retained observable space. Their method uses dominant
modes for prediction. We use the mode--residue pairs to explain bias in an
inferred generator. Eigenvalues alone are insufficient because a slowly
decaying mode with a small residue can have little effect on the recovered
rates.

The same calculus separates the size of the projection defect from the
conditioning of the logarithm. Suppose a contour $\Xi$ surrounds the spectra
of $\bK_h$ and $e^{h\bA_{JJ}}$ without crossing $(-\infty,0]$, and the strip
condition above holds. Then
\begin{equation}\label{eq:log_resolvent_identity}
 \log\bK_h-h\bA_{JJ}
 =\frac{1}{2\pi\mathrm i}\int_\Xi \log z\,
 (zI-\bK_h)^{-1}(\bK_h-e^{h\bA_{JJ}})
 (zI-e^{h\bA_{JJ}})^{-1}\,dz.
\end{equation}
We take norms in \eqref{eq:log_resolvent_identity} to bound the rate error by the projection defect, amplified by the two resolvents and by $1/h$ after applying $\mathcal R$. The contour can cease to exist when the projected spectrum reaches the principal-logarithm cut. The expansion below complements this global spectral statement by resolving the first small-$h$ contribution that survives in the selected species rows.

\subsection{Added observables and bias order}

To remove these paths through a chosen order, we retain $\calA x_s,\ldots,\calA^r x_s$ for each species used in rate extraction. Only these species rows must close; the full dictionary need not be invariant. The first surviving rate error is determined by one excursion from the retained space into its complement and back.

\begin{defn}[Species closure depth]\label{def:closure_depth}
The dictionary space $\calF$ has \emph{closure depth $r$} on a collection of species observables if
\begin{equation}\label{eq:closure_depth}
 \calA^j x_s\in\calF,\qquad j=1,\ldots,r,
\end{equation}
for every species used by the rate-extraction map.
\end{defn}

\begin{theorem}[Rate bias from finite-depth closure]\label{thm:generator_depth}
Let the state space be finite and closed, and let a full-support measure independent of $h$ induce $P$. Suppose that $r\geq1$ and
\[
 \calA^j x_s\in\calF,\qquad j=1,\ldots,r,
\]
for every species used by the extraction map. If $h\|\bA_{\rm full}\|<\log 2$, then
\begin{equation}
 c_h-c
 =-\frac{B_{r+1}}{(r+1)!}h^{r+1}
  \mathcal R(\bA_{JJ}^{r}\bA_{JJ'}\bA_{J'J})
 +O(h^{r+2}).\label{eq:generator_depth_order}
\end{equation}
where we let $B_n$ denote the Bernoulli numbers satisfying $z/(e^z-1)=\sum_{n=0}^{\infty}B_nz^n/n!$. Consequently,
\begin{equation}\label{eq:depth_parity_order}
 c_h-c=
 \begin{cases}
  O(h^{r+1}),&r\text{ odd},\\
  O(h^{r+2}),&r\text{ even}.
 \end{cases}
\end{equation}
For odd $r$, the order is exactly $r+1$ if and only if $\mathcal R(\bA_{JJ}^{r}\bA_{JJ'}\bA_{J'J})\ne0$.
\end{theorem}
\begin{proof}
The closure conditions imply
\begin{equation}\label{eq:depth_block_zero}
 [\bA_{JJ}^{j-1}\bA_{JJ'}]_{\rm sp}=0,
 \qquad j=1,\ldots,r.
\end{equation}
Indeed, after $j-1$ generator applications the selected species rows have no complementary component; after one further application, we obtain the expression in \eqref{eq:depth_block_zero}. Any term that leaves the retained space must later return through $\bA_{J'J}$. The first such word that can survive in a selected species row is therefore
\begin{equation}\label{eq:first_excursion_product}
 [\bA_{JJ}^{r}\bA_{JJ'}\bA_{J'J}]_{\rm sp}.
\end{equation}

Insert the exponential expansion of Lemma~\ref{lem:projected_operator} into the ordinary power series for $\log\bK_h$ and collect this word. At its first possible degree, terms containing $\bA_{J'J'}$, a retained factor to the right of $\bA_{J'J}$, or a second excursion cannot contribute. The remaining coefficients are generated by
\begin{equation}\label{eq:bernoulli_generating_function}
 \frac1z-\frac1{e^z-1}
 =\frac12-\frac{z}{12}+\frac{z^3}{720}-\cdots.
\end{equation}
Hence the coefficient of the word in \eqref{eq:first_excursion_product} is
$-B_{r+1}/(r+1)!$. Dividing by $h$ and applying $\mathcal R$ gives
\eqref{eq:generator_depth_order}. Under the condition
$h\|\bA_{\rm full}\|<\log 2$, the logarithm series converges and the
remainder is uniform. Finally, $B_{r+1}=0$ for positive even $r$, which
proves~\eqref{eq:depth_parity_order}.
\end{proof}

\begin{remark}[Closure as spectral decoupling]
Under condition~\eqref{eq:depth_block_zero}, every short block walk vanishes
if it starts in an extracted species row and enters the omitted space during
the first $r$ generator applications. The full Koopman spectrum is unchanged,
but the retained--omitted decomposition prevents the modal residues in
Proposition~\ref{prop:modal_residues} from reaching the extracted rows along
these short paths. Theorem~\ref{thm:generator_depth} gives the resulting delay
in the rate bias.
\end{remark}

\begin{corollary}[Depth-one projection bias]\label{prop:projection_log}
If $\calA x_s\in\calF$ for every species used by the extraction map, then
\begin{equation}
 c_h-c
 =-\frac{h^2}{12}
 \mathcal R(\bA_{JJ}\bA_{JJ'}\bA_{J'J})+O(h^3).
 \label{eq:rate_block_bias}
\end{equation}
\end{corollary}
\begin{proof}
Let $r=1$ in Theorem~\ref{thm:generator_depth} and use $B_2=1/6$.
\end{proof}

The specified reaction channels generate a space with the closure required by
Definition~\ref{def:closure_depth}. Let $\calF_0$ be a base dictionary that
contains the constant, the species, and the propensity functions needed for
extraction, and let
\begin{equation}\label{eq:channel_product_space}
 \calF_r
 =\calF_0+\spn\{\calA_{k_j}\cdots\calA_{k_1}x_s:
       1\leq j\leq r,\ 1\leq k_1,\ldots,k_j\leq M,\
       s\text{ is used by }\mathcal R\}.
\end{equation}
Zero functions and linear combinations already present are discarded when the dictionary is assembled.

\begin{corollary}[Closure under reaction operators]\label{cor:channel_closure}
For every rate vector $c$, the space $\calF_r$ has closure depth $r$ on the selected species. Consequently, its rate bias has the order in \eqref{eq:depth_parity_order}.
\end{corollary}
\begin{proof}
We use $\calA=\sum_{k=1}^M c_k\calA_k$ repeatedly. For each $j=1,\ldots,r$ and each selected species $x_s$, we expand
\begin{equation}\label{eq:channel_product_expansion}
 \calA^j x_s
 =\sum_{k_1,\ldots,k_j=1}^M
   c_{k_j}\cdots c_{k_1}
   \calA_{k_j}\cdots\calA_{k_1}x_s.
\end{equation}
Every function on the right belongs to the span in \eqref{eq:channel_product_space}. Hence $\calA^j x_s\in\calF_r$ for every $j\leq r$, and the space has closure depth $r$. We now apply Theorem~\ref{thm:generator_depth} to obtain the stated bias order.
\end{proof}

\begin{remark}[Partial closure]\label{rem:partial_closure}
The order theorem requires the complete space in \eqref{eq:channel_product_space}. The Lotka--Volterra and Brusselator comparisons use this complete depth-two space. The Michaelis--Menten and MAPK comparisons retain selected directions from it. Under depth-one closure, the complementary coefficients of $\calA^2x_s$ are $[\bA_{JJ}\bA_{JJ'}]_{\rm sp}$. The leading bias therefore satisfies
\[
 \bigl\|\mathcal R(\bA_{JJ}\bA_{JJ'}\bA_{J'J})\bigr\|
 \leq \|\mathcal R\|\,
 \|[\bA_{JJ}\bA_{JJ'}]_{\rm sp}\|\,\|\bA_{J'J}\|.
\]
Adding directions from the residuals of $\calA^2x_s$ reduces the middle
factor. A rate-independent partial-closure rule orders the empirical residual
directions by singular value and retains the largest observed $L^2$
components first. The final factor and the extraction map also enter the
bound, so partial closure need not reduce every rate error monotonically.
\end{remark}

\subsection{Constructing the added observables}
We construct the spaces in~\eqref{eq:channel_product_space} one depth at a time. At each depth, we apply every unit-rate reaction operator to the functions added at the preceding depth and discard zero functions and linear combinations already represented by the dictionary. Before this reduction, there are $M^j$ ordered sequences of $j$ reaction operators. Different sequences may produce the same polynomial or a function already represented by the dictionary, so several sequences can collapse to one coordinate.

For the reversible dimerization model used below, let channel 1 be $2X\to Y$ and channel 2 be $Y\to2X$. Starting from $\{1,Y,\binom X2\}$, the functions $\calA_1Y$ and $\calA_2Y$ already lie in the base dictionary. At length two,
\[
 \calA_1\calA_1Y=\binom X2(3-2X),\qquad
 \calA_2\calA_1Y=Y(2X+1),\qquad
 \calA_1\calA_2Y=-\binom X2,\qquad
 \calA_2\calA_2Y=Y.
\]
Using the conservation law $X+2Y=24$, the last three functions lie in the base dictionary. We add the component of $\binom X2(3-2X)$ that is independent of the existing functions. One new coordinate is enough to enforce depth-two closure in this example.

From Theorem~\ref{thm:generator_depth}, we obtain orders of at least $2,4,4,$ and $6$ at depths one through four. The repeated fourth order has a different explanation at the two middle depths: the depth-two Bernoulli coefficient is zero, while the depth-three coefficient is $1/720$. The depth-three function can therefore reduce the fourth-order coefficient without raising the order. For polynomial libraries, each additional closure level raises the required degree by at most $\max_k\deg\phi_{\bnu_k}-1$. The Brusselator has a trimolecular channel, so a full polynomial dictionary that enforces depth two must include terms through degree five.

\section{Numerical Validation}
\label{sec:results}
% =============================================================================

The numerical experiments separate the three errors that govern rate
recovery. Exact finite-state matrices isolate projection bias and verify its
predicted order. Simulated trajectories then show how sampling error,
projection bias, and the logarithm branch condition interact in finite data.

\subsection{Experimental setup}

The dimerization model resolves finite-time bias by closure depth and omitted
Koopman mode. A sweep over observation intervals for the Brusselator follows
rate recovery to the principal-logarithm branch cut. The final comparison
tests complete or selected reaction-generated closure on oscillatory,
conserved, higher-order, and multi-tier networks.

The dimerization state space is small enough to enumerate. Lotka--Volterra supplies an oscillatory network; Michaelis--Menten tests conservation reduction and shared propensities; the Brusselator contains a third-order propensity; and a reduced MAPK cascade supplies a signaling network with a large conservation class.

For each network, we fix the species coefficients used to extract its rates and retain the constant, the physical scaling, and the full-rank EDMD fit as we change dictionaries. The Michaelis--Menten conservation class is finite. Lotka--Volterra and the Brusselator have unbounded natural state spaces and therefore test the method beyond the setting of the finite-state bias theorem.

\subsection{Exact finite-state validation}
\label{sec:dimerization}

Consider the reversible dimerization network
\[
 2X\xrightarrow{c_1}Y,\qquad
 Y\xrightarrow{c_2}2X,\qquad c_1=0.025,\quad c_2=0.7.
\]
Using the conservation law $X+2Y=24$, we reduce the model to 13 states and enumerate them to form the exact Gram matrices from the full matrix exponential. We denote the initial probability of state $i$ by $\omega_i$ and choose these probabilities proportional to $1,\ldots,13$ in increasing $Y$ order. The dictionary $\bpsi=(1,Y,\binom X2)^\top$ contains $\calA Y=c_1\binom X2-c_2Y$. Repeated generator applications leave this dictionary. No simulated data enter the resulting bias calculation.

Figure~\ref{fig:dimerization_modal_residues} resolves the projection error by
omitted Koopman mode. We diagonalize the complementary block in the weighted
orthogonal decomposition used by the population EDMD calculation. At
$h=0.1$, the largest residue belongs to
$\mu=-8.38\mathbin{\pm}0.30\mathrm i$. Decoupling this conjugate pair changes
the inferred relative-rate vector by $4.84\mathrm{e}{-4}$. The slower pair
$-4.75\mathbin{\pm}1.09\mathrm i$ has a smaller residue and changes the vector
by $3.07\mathrm{e}{-4}$, while the modes near $-10.8$ and $-11.8$ have effects
below $10^{-6}$. Decay rate alone therefore does not rank the modes that
matter for inference. Their coupling residues are also needed. The
leave-one-mode-out changes are diagnostic rather than additive because the
matrix logarithm couples all modes.

\begin{figure}[H]
\centering
\includegraphics[width=0.88\textwidth]{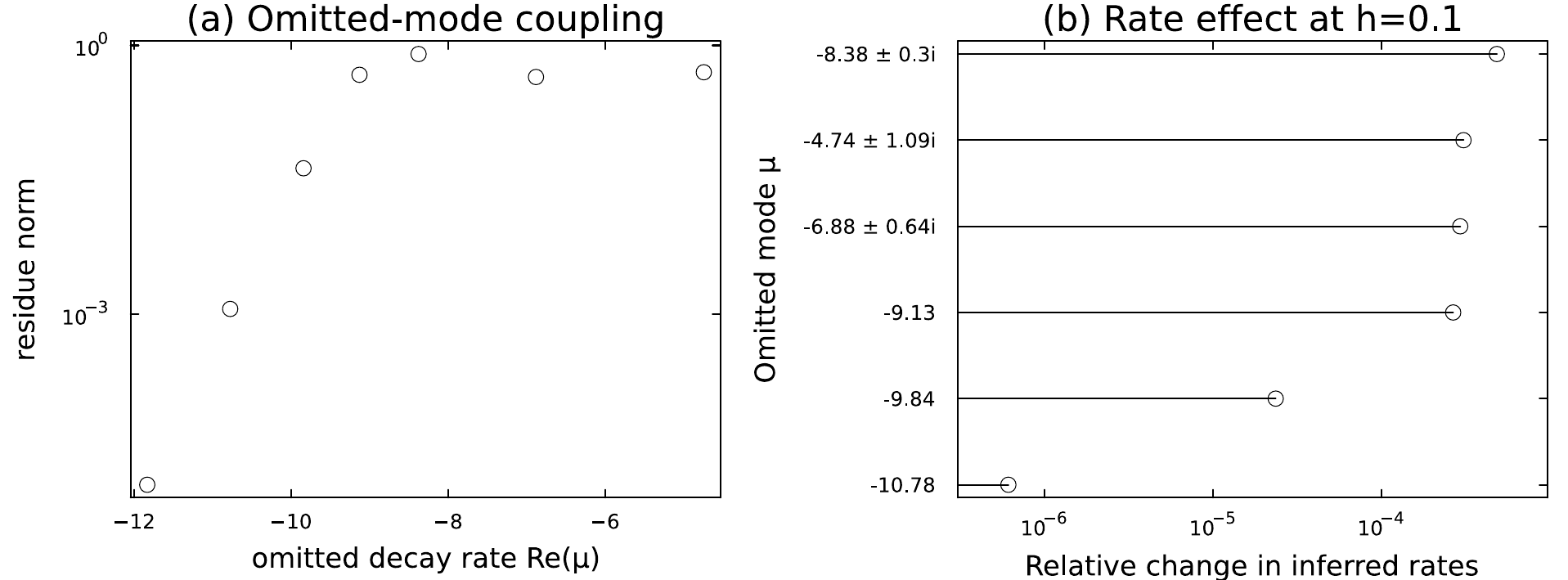}
\caption{Omitted Koopman modes in the exact dimerization calculation. (a)
Eigenvalues of the complementary generator block plotted against the
Frobenius norm of their coupling residues $\bm C_\ell$. Conjugate pairs are
grouped. (b) Change in the inferred relative-rate vector after decoupling one
mode or conjugate pair at $h=0.1$. This leave-one-mode-out quantity measures
influence through the nonlinear logarithm and is not an additive decomposition
of the total bias.}
\label{fig:dimerization_modal_residues}
\end{figure}

\subsubsection{Finite-depth generator closure}
\label{sec:krylov_enrichment}
We test the order theorem without truncation or Monte Carlo error in the finite dimerization model by adding $\calA^j x_s$ to the dictionary for successive values of $j$. Starting with $\{1,Y,\binom{X}{2}\}$, we add the weighted-orthogonal residuals of $\calA^2Y,\calA^3Y$, and $\calA^4Y$ one at a time. We call the largest retained power the closure depth. For this network, the same nested spaces are obtained by successively applying the reaction operators in Corollary~\ref{cor:channel_closure} and removing linearly dependent functions.

\paragraph{Order verification}
In Figure~\ref{fig:krylov_depth_population}(a), we plot the infinite-sample bias across all four depths, for which the fitted slopes are $1.996$, $3.994$, $4.084$, and $5.994$. These values agree with the orders $2,4,4,$ and $6$ in Theorem~\ref{thm:generator_depth}, so we conclude that its lower bounds are attained in this model. At depth three, we use the theorem to obtain the fourth-order coefficient $(1/720)\mathcal R(\bA_{JJ}^3\bA_{JJ'}\bA_{J'J})$. For the two rates, its values are $-3.531\mathrm{e}{-6}$ and $1.053\mathrm{e}{-4}$; by directly evaluating $(c_h-c)/h^4$ at $h=5.00\mathrm{e}{-4}$, we obtain $-3.544\mathrm{e}{-6}$ and $1.056\mathrm{e}{-4}$. The depth-three function therefore changes the fourth-order coefficient, and the depth-four Bernoulli coefficient vanishes.

\paragraph{Bias and sampling error}
An enlarged dictionary has lower first-order mean squared error when its reduction in squared projection bias exceeds its increase in sampling variance. We compare closure depths using a common measure of physical rate error. For $M$ rates, let $v_k(h)$ be the single-unit asymptotic variance coefficient for rate $k$ obtained from the linearized estimator. With $m$ independent sampling units, the first-order relative mean squared error is
\begin{equation}\label{eq:mse_criterion}
 \operatorname{MSE}(h,m)
 \approx \frac1M\sum_{k=1}^M
 \left(\frac{c_{h,k}-c_k}{c_k}\right)^2
 +\frac1{mM}\sum_{k=1}^M
 \frac{v_k(h)}{c_k^2}.
\end{equation}
The first sum is the squared projection bias. The second is the sampling variance, with $v_k(h)$ predicted by the Gram-matrix perturbation calculation in~\eqref{eq:compact_sampling_chain} for one sampling unit and then divided by $m$. The sampling unit is an independent pair, an independent triple, or a whole trajectory, according to the estimator and the way observations are collected.

The short-interval variance following
\eqref{eq:compact_sampling_chain} explains why closure should be compared
after choosing an interval for each dictionary. The following corollary
separates this effect from the fixed-interval comparison below.
\begin{corollary}[Optimized error rate]\label{cor:optimal_interval}
Suppose that the aggregate relative bias is $\beta h^p+o(h^p)$ and the single-pair relative variance coefficient is $\nu/h+O(1)$, where $\beta,\nu>0$. Then the leading mean squared error is minimized at
\begin{equation}\label{eq:optimal_interval}
 h_*(m)=\left(\frac{\nu}{2p\beta^2m}\right)^{1/(2p+1)},
\end{equation}
and its minimum equals
\begin{equation}\label{eq:optimal_mse_rate}
 (2p+1)\beta^2
 \left(\frac{\nu}{2p\beta^2m}\right)^{2p/(2p+1)}.
\end{equation}
Thus the ordinary depth-one logarithm has the rate $m^{-4/5}$, while depth-two closure has the rate $m^{-8/9}$ when its fourth-order coefficient is nonzero.
\end{corollary}
\begin{proof}
The leading terms are $\beta^2h^{2p}+\nu/(mh)$. We differentiate to obtain $2p\beta^2h^{2p+1}=\nu/m$ and hence~\eqref{eq:optimal_interval}. At this value, the variance term is $2p$ times the squared-bias term; their sum is~\eqref{eq:optimal_mse_rate}.
\end{proof}

\paragraph{Optimized observation interval}
In Figure~\ref{fig:krylov_depth_population}(b), we minimize the exact first-order mean squared error over $2.00\mathrm{e}{-4}\leq h\leq0.5$ separately for each depth and sample size. Depth two first improves on depth one at approximately $1.00\mathrm{e}{6}$ pairs, because its smaller bias permits an interval near $0.3$ while the depth-one optimum decreases with the sample size. At $1.00\mathrm{e}{10}$ pairs, the optimized depth-two mean squared error is $26.0\%$ of the optimized depth-one value. Depths three and four have larger sampling coefficients and do not improve on depth two over the plotted range.

\begin{figure}[htbp]
\centering
\includegraphics[width=\textwidth]{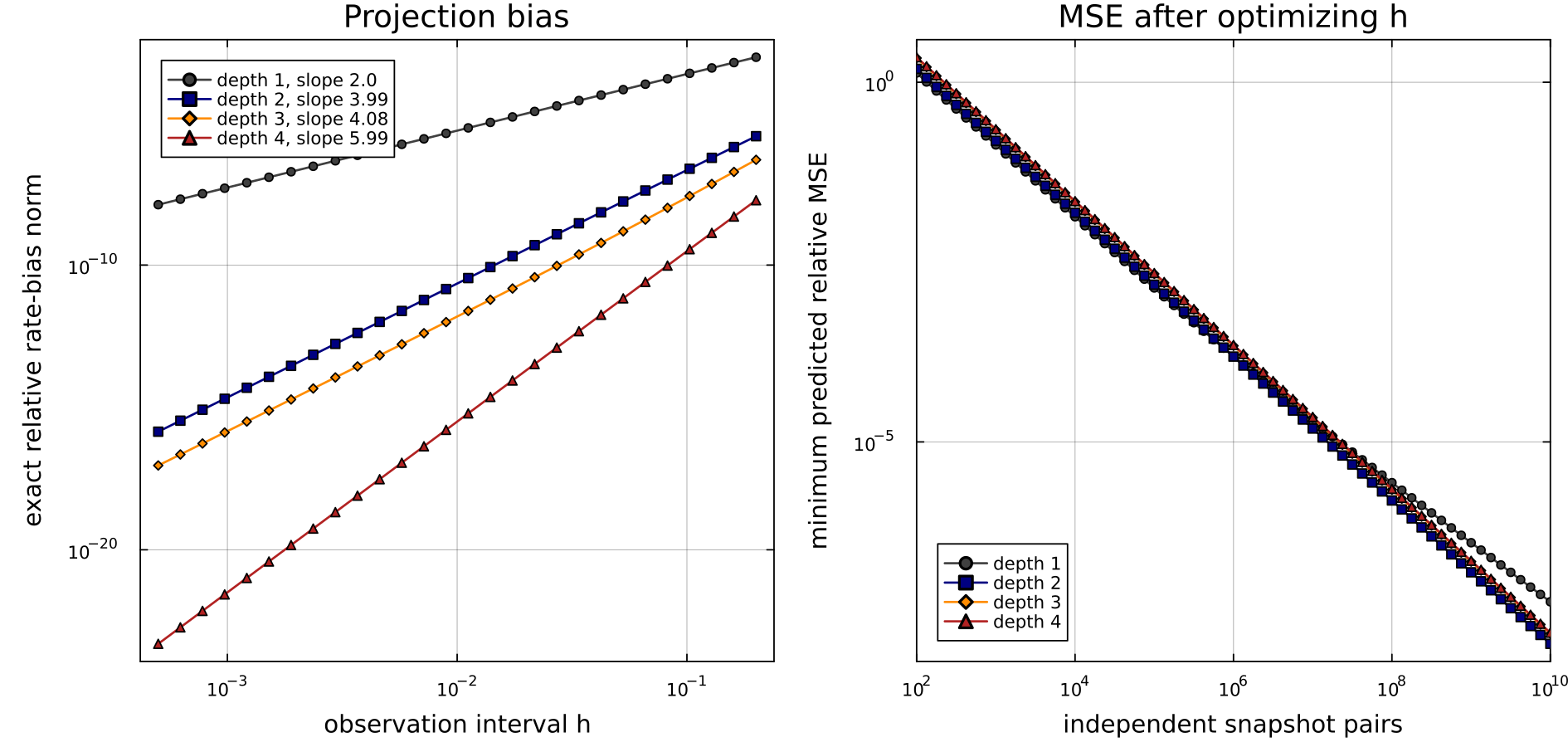}
\caption{Finite-depth generator closure in the exactly enumerable dimerization model. (a) Exact relative rate bias, computed with 256-bit arithmetic. The dictionary at depth $r$ contains $\calA Y,\ldots,\calA^rY$. (b) First-order relative mean squared error after minimizing over the observation interval for each depth and sample size; its bias contribution is evaluated from the exact projected operators.}
\label{fig:krylov_depth_population}
\end{figure}

\subsection{Brusselator over a wide range of observation intervals}
\label{sec:bru}

Exact enumeration separated projection bias from sampling error in the
dimerization study. We now use simulated Brusselator trajectories to examine
the practical choice of observation interval. The experiment begins in the
accurate short-interval regime, passes through projection-dominated error,
and ends where the fitted logarithm is no longer admissible.

The Brusselator~\cite{prigogine1968symmetry} combines a constant source,
autocatalysis, conversion, and degradation. It tests source-rate recovery and
the separation of two reactions with the same linear $X$ propensity. The
volume-scaled stochastic model with $\Omega=200$ is
\[
\varnothing\xrightarrow{200}X,\quad
2X+Y\xrightarrow{2.50\mathrm{e}{-5}}3X,\quad
X\xrightarrow{3.0}Y+D,\quad
X\xrightarrow{1.0}E,
\]
which starts from $\bx_0=(200,600)$ and runs to $T=50$ with $\Delta t=0.1$, while only $(X,Y)$ are retained as dynamical species. The ten-dimensional mass-action basis retains the constant and contains all factorial monomials through degree three. Let $\phi(x,y)=x(x-1)y/2$ denote the trimolecular propensity feature. The effective stochastic rate constants are $c_1=200$, $c_2=2.50\mathrm{e}{-5}$, $c_3=3.0$, and $c_4=1.0$; any chemostatted factor is absorbed into $c_3$. Applying Lemma~\ref{clm:species_action}, we obtain the species-row entries
\begin{align*}
(\bA)_{X,x}&=-(c_3+c_4)=-4.0,&
(\bA)_{X,\phi}&=c_2=2.50\mathrm{e}{-5},\\
(\bA)_{Y,x}&=c_3=3.0,&
(\bA)_{Y,\phi}&=-c_2.
\end{align*}
We add the species actions to cancel both autocatalysis and conversion:
\begin{equation}\label{eq:bru_sum}
 \calA(X+Y)=c_1-c_4X.
\end{equation}
Hence $c_4=-(\bA)_{X,x}-(\bA)_{Y,x}$ is recovered from the coefficient of
$X$. The controlled studies below use this relation together with the
individual $c_3$ and autocatalysis coefficients.

\label{sec:bru_controlled}
The controlled calculation uses 1000 independent direct-SSA trajectories (seed 714) with the stated propensities and initial state. We retain the constant and all ten factorial features through degree three. At $h=0.1$, the bootstrap standard deviations for $c_3,c_4$ are $0.0696,0.0188$, while the linearized predictions are $0.0697,0.0187$. The estimate $\widehat c_3=4.524$ differs substantially from the true value $3$. We vary $h$ and use an independent reference ensemble to separate sampling variation from this systematic error.

\label{sec:lag_regimes}
The extended calculation covers 14 intervals from $5.00\mathrm{e}{-4}$ to $6.4$, with every fit using the same factorial dictionary and rate-extraction formulas. We retain the intercept and all nonzero singular values and require a real principal logarithm. All intervals share current times $0:0.1:49.6$ and 100 paired whole-trajectory bootstrap index sets. Simulation extends to time $56$ to supply the longest intervals. The source and autocatalytic rates are now included, with autocatalysis extracted by averaging the two opposite species coefficients.

An independent ensemble of $10{,}000$ trajectories (seed 2718), divided into
twenty groups, provides the empirical reference. Variation among the groups
gives linearized Monte Carlo standard errors. The reference therefore retains
quantified Monte Carlo uncertainty, while every interval in the study uses
the same trajectories and sampling mixture. Figure~\ref{fig:lag_recovery}
reports all four rates together with the bootstrap and reference
uncertainties.

\begin{figure}[htbp]
\centering
\includegraphics[width=\textwidth]{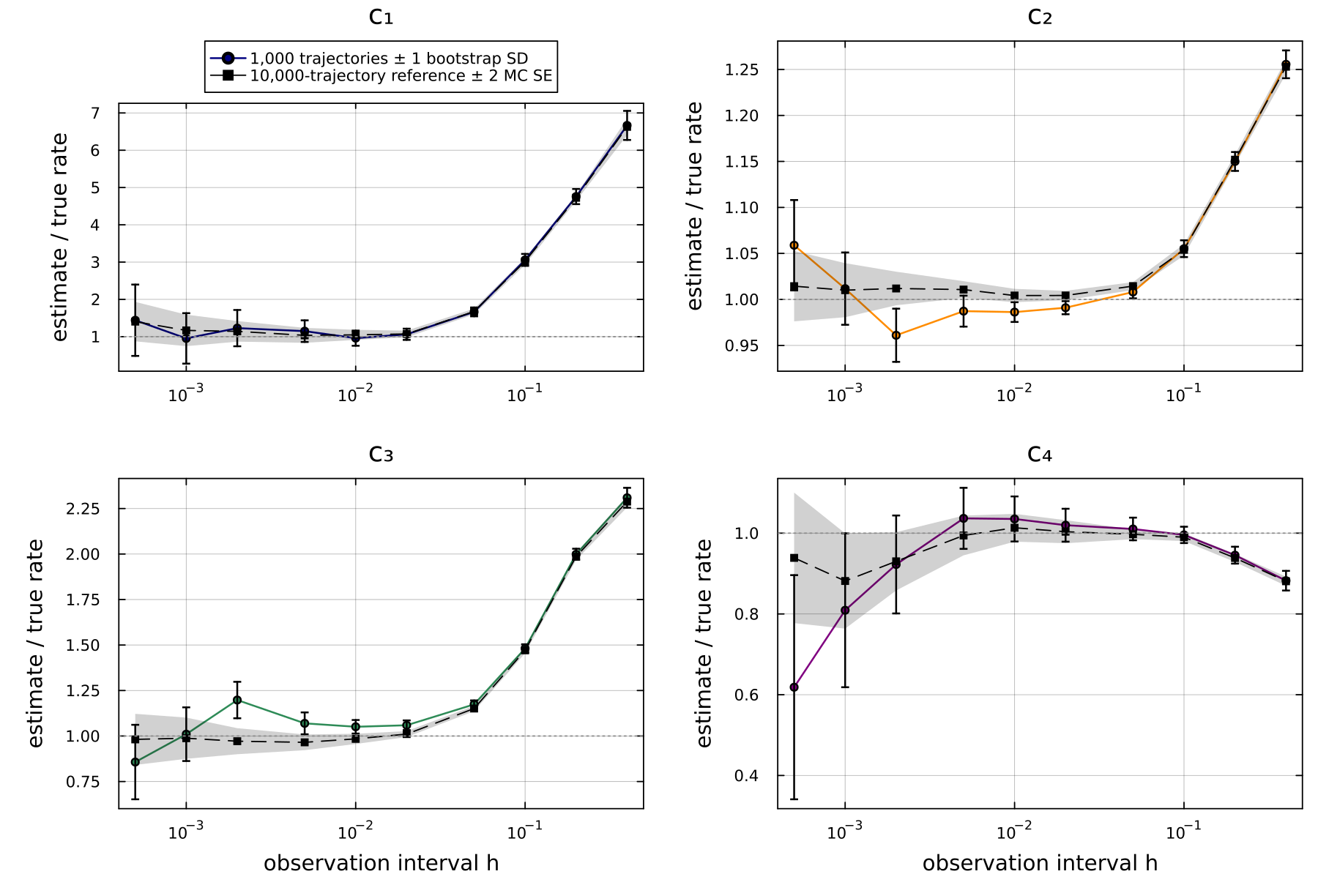}
\caption{Rate recovery across observation intervals, normalized by each true rate. We plot one bootstrap standard deviation for 1000 trajectories as error bars and two Monte Carlo standard errors for the independent $10{,}000$-trajectory reference as gray bands. The plotted range ends at $h=0.4$, the largest tested interval for which the central fitted matrix has a real principal logarithm.}
\label{fig:lag_recovery}
\end{figure}

In the figure, we distinguish three regimes. At short intervals, the bootstrap standard deviations for $c_3$ and $c_4$ fall from $0.614$ and $0.278$ at $h=5.00\mathrm{e}{-4}$ to $0.112$ and $0.0558$ at $h=0.01$, while the order-$h^2$ projection bias remains below the sampling resolution. At intermediate intervals, bias becomes rate dependent: increasing $h$ from $0.02$ to $0.1$ leaves $c_4$ near one but moves the conversion estimate from $3.177$ to $4.443$. From the independent reference, we obtain $4.413\pm0.024$ at $h=0.1$, confirming that this discrepancy is systematic. At $h=0.8$, $1.6$, and $3.2$, every bootstrap fit and the reference fit fail the real principal-logarithm check. Thus sampling error dominates at short intervals, projection bias dominates selected rates at intermediate intervals, and the fitted logarithm becomes inadmissible at long intervals.

Figure~\ref{fig:bru_spectrum} gives the corresponding spectral view. Define
the branch margin
\begin{equation}\label{eq:branch_margin}
 \delta_h=\min_{x\leq0}\sigma_{\min}(\bK_h-xI),
\end{equation}
This quantity measures distance to the negative real axis in a
singular-value sense and therefore accounts for nonnormality. For the central
fit, $\delta_h$ falls from $1.12\mathrm{e}{-2}$ at $h=0.1$ to
$2.39\mathrm{e}{-4}$ at $h=0.4$. It vanishes at $h=0.8$ when an eigenvalue
reaches the branch cut. The independent reference gives nearly identical
values, so the loss of a real logarithm is reproducible across the two
ensembles.

\begin{figure}[htbp]
\centering
\includegraphics[width=0.92\textwidth]{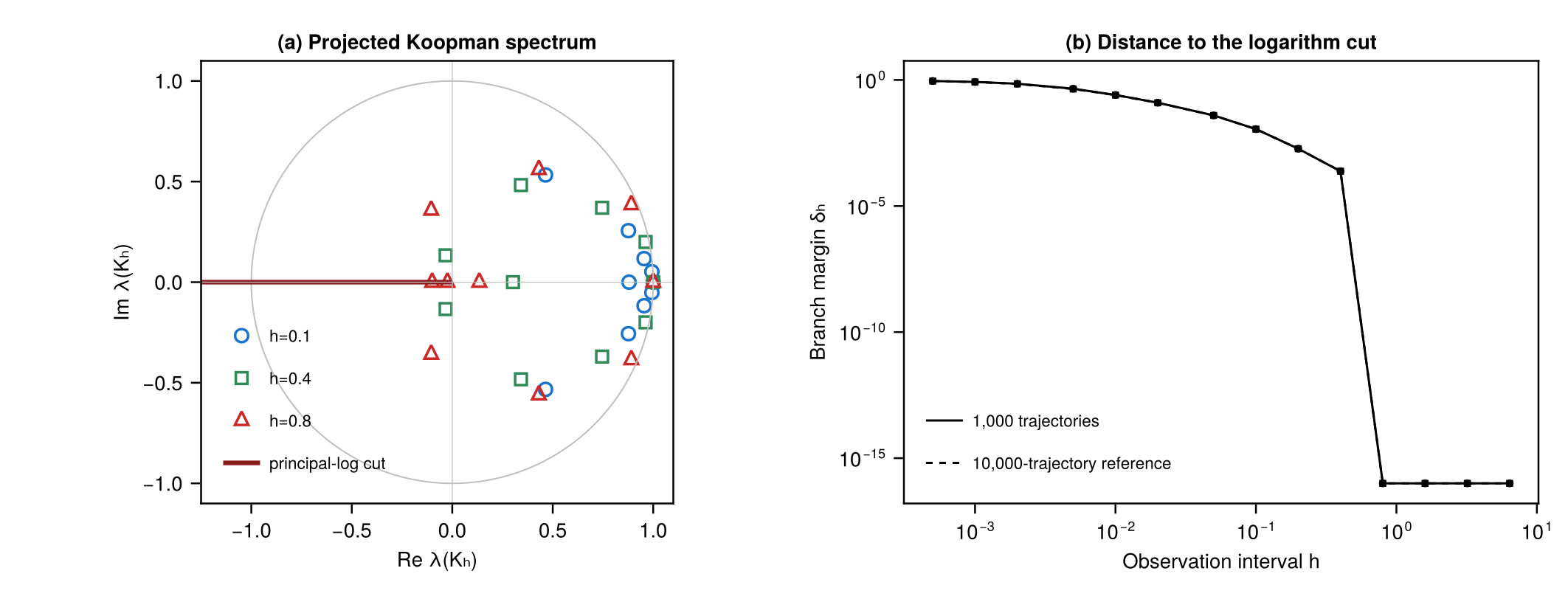}
\caption{Spectral diagnostics for the projected Brusselator Koopman matrices. (a) Eigenvalues of the central 1000-trajectory fit at three observation intervals; the red segment is the principal-logarithm branch cut and the gray circle is the unit circle. (b) The branch margin \eqref{eq:branch_margin} for the central and independent reference ensembles. Values plotted at machine scale are zero to working precision because an eigenvalue lies on the cut.}
\label{fig:bru_spectrum}
\end{figure}

With the current sample sizes, the asymptotic bias slope remains unresolved at the shortest intervals.

\subsection{Comparison across network classes}
\label{sec:extended_closure_benchmark}

The dimerization and Brusselator studies isolate bias order, sampling cost,
and interval selection. We next test a small number of added observables on
four reaction networks with different structures. The additional functions
do not depend on the unknown rates. For every channel pair $(\ell,k)$, we
evaluate
$\phi_{\bnu_\ell}(\bx)
[\phi_{\bnu_k}(\bx+\bzeta_\ell)-\phi_{\bnu_k}(\bx)]$. A function is discarded
if it vanishes on all current observations. We subtract the empirical
projection of each remaining function onto the base dictionary and compute an
SVD of the residuals. Independent residuals are added in decreasing order of
singular value, while the base coordinates remain unchanged. The resulting
dictionaries are nested, beginning with the ordinary estimator and ending
with full depth-two closure.

Separate exploratory ensembles choose the observation interval and the number
of added functions. Among the admissible settings considered, we select the
largest paired reduction in mean absolute relative rate error. This criterion
uses one common interval for the base and closure estimators. We then fix the
interval and number of functions before generating twenty new evaluation
datasets for each network. The base and closure fits use identical
observations within each dataset.

The largest model in this comparison is a reduced Huang--Ferrell MAPK cascade~\cite{huang1996ultrasensitivity}, with one phosphorylation at the upper tier and two at each of the lower tiers. After omitting explicit enzyme--substrate complexes, the model has ten mass-action channels: five phosphorylations catalyzed by the stimulus or the active species in the tier above, and five first-order dephosphorylations. The conserved tier totals are $40$, $60$, and $60$, and twelve initial conditions vary the activity of each tier separately while including several intermediate activation states. After the reduction, the model has six independent dynamical coordinates and a 12-function base dictionary. At a fixed stimulus level, its conservation class still contains $41\binom{62}{2}^2=146{,}611{,}121$ states.

Table~\ref{tab:extended_closure_benchmark} compares the two dictionaries on
the same observations. Within each dataset, we first average the absolute
relative errors over all rates. The error columns then report the median and
interquartile range across twenty datasets. The condition numbers are the
corresponding medians for the scaled empirical Gram matrices. Every projected
matrix has a real principal logarithm. The Lotka--Volterra and Brusselator
rows use overlapping intervals within each trajectory. Each
Michaelis--Menten and MAPK pair comes from an independent short trajectory.
We draw an initial condition and a current time from their fixed joint
distribution, simulate to that time, and then continue once over the
observation interval. Exploratory seeds 8101 and 9101--9501 select the
intervals and added functions. The reported seeds $10101$--$12001$, in
increments of $100$, are reserved for evaluation.

\begin{table}[htbp]
\centering
\scriptsize \caption{Independent comparison of observable sets using the same observations. Errors are percentages. The dimension column lists the base, selected, and complete depth-two dimensions. Added coordinates are orthogonalized against the scaled base dictionary. Their Gram eigenvalues lie within the extremal eigenvalues of the base Gram matrix, so the displayed condition numbers are unchanged. Intervals and selected dimensions were fixed on exploratory data before the twenty evaluation datasets were generated.}
\label{tab:extended_closure_benchmark}
{\setlength{\tabcolsep}{4.5pt}
\begin{tabular}{lrrrrrr}
\toprule
Network & $h$ & Pairs & $d_0/d_{\rm sel}/d_2$ & Base error & Closure error &
Condition numbers\\
\midrule
Lotka--Volterra & 0.2 & $100{,}000$ & $6/7/7$ &
$0.936\;(0.653,1.05)$ & $0.767\;(0.378,1.02)$ &
$2.84\mathrm{e}{3}\to2.84\mathrm{e}{3}$\\
Michaelis--Menten & 6.0 & $64{,}000$ & $4/5/6$ &
$5.56\;(5.22,6.13)$ & $2.79\;(1.50,4.19)$ &
$8.48\mathrm{e}{2}\to8.48\mathrm{e}{2}$\\
Brusselator & 0.04 & $248{,}500$ & $10/11/11$ &
$15.9\;(13.7,17.9)$ & $4.88\;(2.61,6.04)$ &
$3.58\mathrm{e}{6}\to3.58\mathrm{e}{6}$\\
MAPK cascade & 1.6 & $96{,}000$ & $12/15/21$ &
$4.90\;(3.92,5.60)$ & $2.90\;(2.59,3.67)$ &
$2.46\mathrm{e}{4}\to2.46\mathrm{e}{4}$\\
\bottomrule
\end{tabular}
}
\end{table}

\begin{figure}[htbp]
\centering
\includegraphics[width=\textwidth]{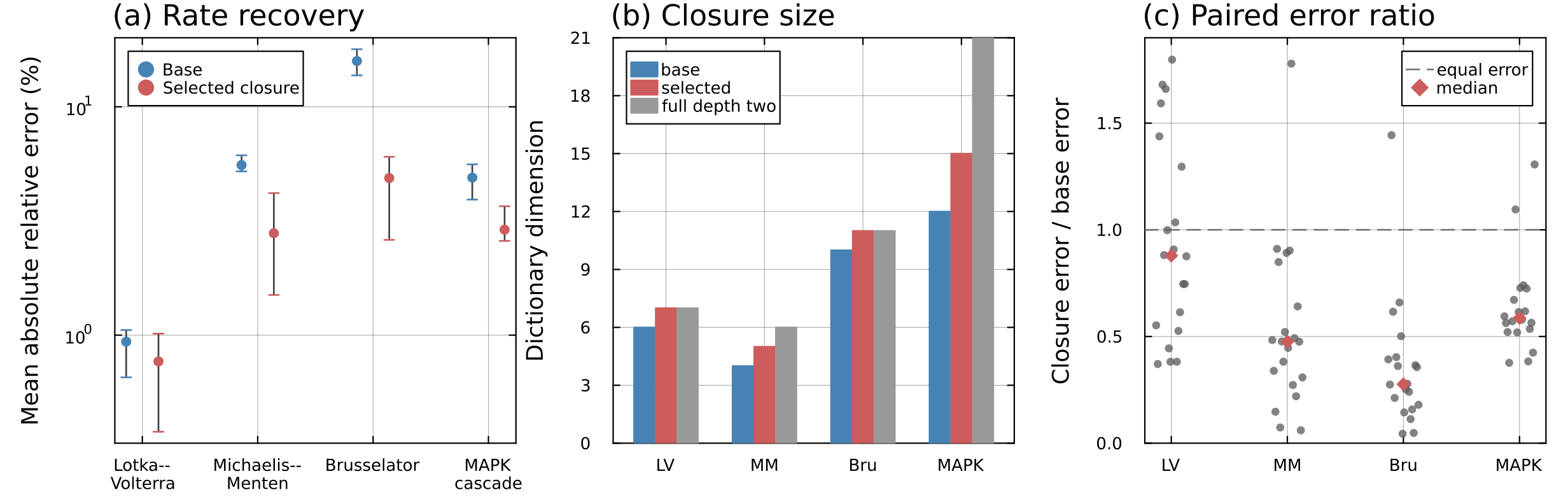}
\caption{Independent comparison after adding selected reaction-generated observables. (a) Median mean absolute relative rate error across twenty datasets; bars span the interquartile range. (b) Base, selected, and complete depth-two dictionary dimensions. (c) Paired ratio of the closure error to the base error for each dataset; diamonds mark medians and the dashed line marks equal error. Each base and closure fit uses the same observations. All four networks use evaluation seeds that were excluded from interval and dimension selection.}
\label{fig:network_closure_benchmark}
\end{figure}

Figure~\ref{fig:network_closure_benchmark} reports the paired comparison at the selected observation intervals and sample sizes. The median errors decrease from $0.936\%$ to $0.767\%$ for Lotka--Volterra, $5.56\%$ to $2.79\%$ for Michaelis--Menten, $15.9\%$ to $4.88\%$ for the Brusselator, and $4.90\%$ to $2.90\%$ for MAPK. With closure, we lower the paired error in $13$, $19$, $19$, and $17$ of the twenty datasets, respectively. Exact two-sided paired Wilcoxon tests yield $p=0.202$, $1.68\mathrm{e}{-4}$, $3.81\mathrm{e}{-6}$, and $1.34\mathrm{e}{-4}$ in the same order. Thus the Lotka--Volterra reduction is modest and unresolved at this sample size, while we observe clear paired reductions in the other three comparisons. We estimate one and three added functions for the Michaelis--Menten and MAPK calculations, respectively.

The finite Michaelis--Menten conservation class also permits a comparison
without sampling error. At $h=6$, the mean absolute relative error is
$5.23\%$ with the base dictionary and $0.484\%$ with full depth-two closure.
The infinite-sample calculation therefore agrees with the predicted bias
reduction. The new trajectory data show how much of this reduction remains
when one additional function must be estimated from a finite sample.

The MAPK conservation class contains $146{,}611{,}121$ molecular states. We fit a $15\times15$ Koopman matrix whose observables retain the species and propensity coefficients needed to identify all ten rates.

% =============================================================================
\section{Discussion and Conclusions}
\label{sec:discussion}
% =============================================================================

The method recovers reaction rates from paired molecular states without
constructing the full CME transition matrix. Mass-action structure places the
rates in selected rows of the Koopman generator. EDMD and a matrix logarithm
estimate these rows from finite-time data.

Generator paths that leave the retained observable space and return to the
species rows describe the finite-time projection error. Closure through depth
$r$ represents these paths through the order stated in
Theorem~\ref{thm:generator_depth}. The first remaining coefficient is
$-B_{r+1}/(r+1)!$, which gives order $h^{r+1}$ at odd depth and at least order
$h^{r+2}$ at even depth. The block-resolvent formula describes the same error
spectrally. Each omitted mode contributes a pole $\mu_\ell$ and a coupling
residue $\bm C_\ell$. In the dimerization calculation, the residues, rather
than the decay rates alone, rank the modes by their effect on the inferred
rates. Closure suppresses their shortest paths into the extracted species
rows, while the resolvents in \eqref{eq:log_resolvent_identity} quantify the
remaining amplification near the logarithm branch cut.

The finite dimerization calculation verifies these orders without state-space
truncation or Monte Carlo error. It also shows why the observation interval
and closure depth must be selected together. After optimizing the interval,
depth-two closure first improves on the ordinary logarithm at about
$1.00\mathrm{e}{6}$ pairs. At $1.00\mathrm{e}{10}$ pairs, its predicted mean
squared error is $26.0\%$ as large.

The Brusselator sweep reaches the complementary spectral boundary. The branch margin collapses at $h=0.8$, where both the central fit and the independent reference lose a real principal logarithm. At shorter intervals, we use the same calculation to separate sampling-dominated and projection-dominated rate error.

Independent evaluation datasets test the added observables when they must be
estimated from trajectories. Median errors decrease in all four networks,
with clear paired reductions for Michaelis--Menten, the Brusselator, and
MAPK. For Lotka--Volterra, the median gap is smaller than the variation among
datasets, and the paired test is unresolved. This result agrees with the
finite-sample error budget. Reducing projection bias gives only a modest gain
when sampling variation has comparable size.

The observable formulation exchanges statistical efficiency for a
calculation whose dimension equals the number of selected functions. In the
13-state model, the ordinary estimator has a standard deviation $6.05$ times
the Cram\'er--Rao limit attained asymptotically by the correctly specified
full likelihood. In the MAPK calculation, a $15\times15$ observable matrix
replaces a state-space matrix with about $1.47\mathrm{e}{8}$ rows and columns.
These examples show the intended trade-off. The method avoids enumeration of
the molecular state space while retaining a quantitative account of its bias
and sampling error.

\paragraph{Future work}
Extending the finite-state remainder bound to unbounded networks will require generator-domain conditions and moment estimates. Finite-data applications would also benefit from automatic selection or regularization of the generated candidate functions in place of the present exploratory selection. Such extensions should account explicitly for the covariance of overlapping intervals. Long observation intervals require a separate treatment of logarithm branches and admissibility, closely related to the classical embedding problem \cite{kingman1962embedding,israel2001generators,davies2010embeddable, casanellas2023embedding}. We leave measurement error, discovery of missing reactions, and matrix-free logarithm calculations to future work.

% =============================================================================

\section*{Code and data availability}
The code and data used to reproduce every table and figure are available at \url{https://github.com/AdityaDendukuri/cme_koopman}. 

\bibliographystyle{siamplain}
\bibliography{references,koopman_rate_frechet_references}

\end{document}